\documentclass[11pt, reqno]{amsart}       
\usepackage{graphicx}

\usepackage{amsmath}

\usepackage{amssymb}
\usepackage{pifont}

\usepackage[colorlinks=true,
linkcolor=blue,
citecolor=blue,
urlcolor=blue]{hyperref}
\newtheorem{thm}{Theorem}
\newtheorem{lem}{Lemma}
\newtheorem{prop}{Proposition}
\newtheorem{cor}{Corollary}
\theoremstyle{definition}

\newtheorem*{rem}{Remarks}

\newtheorem*{rem2}{Remark}

\newtheorem*{acknowledgement}{Acknowledgement}

\newcommand{\vertiii}[1]{{\left\vert\kern-0.25ex\left\vert\kern-0.25ex\left\vert
		#1 \right\vert\kern-0.25ex\right\vert\kern-0.25ex\right\vert}}

\def \lim   {\text {\rm lim}}

\begin{document}
	
	\title[]{Recoverable states on von-Neumann algebras}
	
	\author{Saptak Bhattacharya}
	
	\address{Indian Statistical Institute\\
		New Delhi 110016\\
		India}
	\email{saptak21r@isid.ac.in}
	
	
	
	
	\begin{abstract}
		Let $(\mathcal{M},\tau)$ and $(\mathcal{N},\tau^{\prime})$ be tracial von-Neumann algebras and let $\phi:\mathcal{M}\to\mathcal{N}$ be a strictly completely positive, trace preserving map. Given a positive, invertible $B\in\mathcal{M}$ with $\tau(B)=1$, a state on $\mathcal{M}$ given by a positive $A\in L^1(\mathcal{M}, \tau)$ is said to be recoverable if $\mathcal{R}(\phi(A))=A$ where $\mathcal{R}$ is the Petz recovery map corresponding to $B$ and $\phi$. In this paper, we study recoverable states and show how an arbitrary state can be made close to a recoverable state via iterates of $\mathcal{R}\circ\phi$. We show that there exists a completely positive, trace preserving map $\psi:\mathcal{M}\to\mathcal{M}$ such that $\psi(A)$ is recoverable for all $A$ and $(\mathcal{R}\circ\phi)^n\to\psi$ in norm as operators on $L^p(\mathcal{M},\tau)$ for all $1\,\textless p\,\textless\infty$, and discuss potential applications to quantum information theory. We also show that this convergence holds strongly in $L^1$. Finally, we prove an interesting decomposition theorem for normal states on $\mathcal{M}$.   
	\end{abstract}
	
	\date{}
	\maketitle
	
	\section{Introduction} We briefly recall some basics on non-commutative $L^p$ spaces and the sandwiched quasi-relative entropies, which we will be using a lot throughout this paper. Let $\mathcal{H}$ be a Hilbert space and let $B(\mathcal{H})$ be the set of all bounded operators on $\mathcal{H}$. Let $\mathcal{M}\subset B(\mathcal{H})$ be a von-Neumann algebra equipped with a normal, faithful, tracial state $\tau:\mathcal{M}\to\mathbb{C}$. The commutant $\mathcal{M}^{\prime}$ is defined as $$\mathcal{M}^{\prime}=\{X\in B(\mathcal{H}):AX=XA\,\,\,\forall\, A\in\mathcal{M}\}.$$ Consider a closed densely defined operator $A:D(A)\to\mathcal{H}$. $A$ is said to be {\it affiliated} to $\mathcal{M}$ if $U^*AU=A$ for all unitaries in the commutant $\mathcal{M}^{\prime}$. Equivalently, let $A=V|A|$ be the polar decomposition of $A$. Then, $A$ is affiliated to $\mathcal{M}$ if and only if the partial isometry $V$ and all the spectral projections of $|A|$ lie in $\mathcal{M}$. 
	\medskip

	Note that the trace $\tau$ naturally induces a family of $p$-norms on $\mathcal{M}$ given by $$||X||_p=[\tau(|X|^p)]^{1/p}$$ for all $p\geq 1$, $X\in\mathcal{M}$. Now, let $A:D(A)\to\mathcal{H}$ be an affiliated operator. Given a Borel $E\subset\mathbb{R}$ let $P(E)$ be the corresponding spectral projection of $|A|$. Write $$|A|=\int_0^{\infty}\lambda\, dP(\lambda).$$ Then the Segal-Dixmier non-commutative $L^p$ space $L^p(\mathcal{M}, \tau)$ (see \cite{sil, nel} for more details) consists of all such $A$ affiliated to $\mathcal{M}$ with $$\int_{0}^{\infty}\lambda^p\, d\tau(P(\lambda))\,\textless\,\infty.$$ The norm $||.||_p$ extends naturally to $L^p(\mathcal{M},\tau)$ by $$||A||_p=[\int_0^{\infty}\lambda^p\, d\tau(P(\lambda))]^{1/p}.$$ This makes $L^p(\mathcal{M},\tau)$ a Banach space, and a Hilbert space for $p=2$.  
	\medskip
	
	Let $\psi:M\to\mathbb{C}$ be a normal state. Then there exists a positive $A\in L^1(\mathcal{M},\tau)$ such that $\tau(A)=1$ and $\psi(X)=\tau(AX)$ for all $X\in\mathcal{M}$. Let $B\in\mathcal{M}$ be positive and invertible with $\tau(B)=1$. Given any positive $A\in L^p(\mathcal{M},\tau)$ with $\tau(A)=1$ the sandwiched quasi-relative entropy $\mathcal{S}_p$ is defined as \[\mathcal{S}_p(A|B)=\tau[(B^{-1/2q}AB^{-1/2q})^p]\label{e1}\tag{1}\] where $\frac{1}{p}+\frac{1}{q}=1$. Note that $\mathcal{S}_1(A|B)=||A||_1$ and $\mathcal{S}_{\infty}(A|B)=||B^{-1/2}AB^{-1/2}||$. Further details on these entropies in the finite dimensional setup can be found in \cite{ml, sb, gw, jen}.
	\medskip
	
	Let $(\mathcal{N},\tau^{\prime})$ be another von-Neumann algebra. A {\it quantum channel} is a completely positive, trace preserving (henceforth called CPTP) map $\phi:\mathcal{M}\to\mathcal{N}$.

	Given a quantum channel $\phi:(M,\tau)\to(N,\tau^{\prime})$ we can extend it continuously to the Hilbert space $L^2(\mathcal{M},\tau)$. Its adjoint $\phi^*$ maps $L^2(\mathcal{N},\tau^{\prime})$ to $L^2(M,\tau)$. $\phi$ is said to be a strict CPTP map if it takes positive invertible elements to postive invertible elements. Let us discuss a couple of examples of such maps.
	\medskip
	
	\begin{enumerate}
		\item Let $\phi:M_2(\mathcal{M})\to M_2(\mathcal{M})$ be defined as \[\phi\begin{pmatrix} A & B\\ C & D\end{pmatrix}=\begin{pmatrix} A & O\\ O & D\end{pmatrix}.\] This is called a {\it pinching}. It is easy to see that this is strictly CPTP.
		\medskip
		
		\item More generally, any {\it unital} channel $\phi:\mathcal{M}\to\mathcal{N}$ is strictly CPTP. This is because if $B\geq\delta I$ for some $\delta\,\textgreater\, 0$, $\phi(B)\geq\delta I$ by unitality and positivity. This class of examples includes, in particular, trace preserving conditional expectations onto von-Neumann subalgebras of $\mathcal{M}$.
		\medskip

	\end{enumerate} 
	
		 Given a strict CPTP map $\phi:(\mathcal{M},\tau)\to(\mathcal{N},\tau^{\prime})$ and a positive, invertible $B\in \mathcal{M}$, define the Petz recovery map $\mathcal{R}:L^2(\mathcal{N},\tau^{\prime})\to L^2(\mathcal{M},\tau)$ by \[\mathcal{R}(Y)=B^{1/2}\phi^*(\phi(B)^{-1/2}Y\phi(B)^{-1/2})B^{1/2}.\label{e2}\tag{2}\] This is extremely important in quantum information theory, chiefly in the context of state recoverability and stability refinements of quantum data processing inequalities (see \cite{pr, hp, ws, mw, jw, fh, fh2}). Later, we will see that this gives an important example of a strict CPTP map which is not unital.

	Note that even though $\phi^*(\phi(B)^{-1/2}Y\phi(B)^{-1/2})$ might be unbounded, multiplication on both sides with the {\it bounded} operator $B^{1/2}$ makes sense in non-commutative $L^2$. 
	\medskip
	
	We now focus on the entropies $\mathcal{S}_p$, as defined in $\eqref{e1}$. It follows from the work of Jen\v{c}ov\'a (\cite{jencova}) that for any strictly CPTP $\phi:\mathcal{M}\to\mathcal{N}$ and $B\in\mathcal{M}$ positive and invertible with $\tau(B)=1$, \[\mathcal{S}_p(\phi(A)|\phi(B))\leq \mathcal{S}_p(A|B)\label{e3}\tag{3}\] for all positive $A\in L^p(\mathcal{M},\tau)$ with $\tau(A)=1$. This is called the {\it data processing inequality}. For $\mathcal{S}_2$, this inequality was initially refined in the finite dimensional setting by Cree and Sorce in \cite{cs}. Later, Gao et. al. \cite{lg} and we \cite{sap} independently refined it further. Recently, in \cite{sap1}, we obtained the universal bound \[\begin{aligned}4[1-F(\psi_{A}|\psi_{\,\mathcal{R}(\phi(A))})]^2&\leq||A-\mathcal{R}\circ\phi(A)||^2_1\\&\leq [\mathcal{S}_2(A|B)-\mathcal{S}_2(\phi(A)|\phi(B))]\end{aligned}\label{e4}\tag{4}\] in the infinite dimensional setting, where $\psi_{A}$ and $\psi_{\mathcal{R}(\phi(A))}$ are the normal states on $\mathcal{M}$ corresponding to $A$ and $\mathcal{R}\circ\phi(A)$ respectively and $F$ is the Uhlmann fidelity defined as the square root of the transition probability in \cite{uhl2, alb, uhal, uhl}. 
	\medskip
	
	Inequality $\eqref{e4}$ gives the first universal recoverability bound in infinite dimensions and has an interesting physical interpretation in quantum information theory. Let us say Alice and Bob decide to communicate among themselves using a channel $\phi$. They fix a key $B$ known to both, and Alice sends her message through a state $A$. Inequality $\eqref{e4}$ quantifies to what extent Bob can recover the original state $A$ from $\phi(A)$ with the Petz map based on the change in relative entropy with respect to $B$.
	\medskip
	
	In this paper, we deal with the above problem from the perspective of the sender, Alice. Ideally, Alice would want to send her message to Bob through a state $A$ which can be perfectly recovered by $\mathcal{R}$, but this cannot be expected in real life. If she sends any state $A$, she does not know for sure whether Bob will be able to recover it well, since she cannot control the entropy loss $\mathcal{S}_2(A|B)-\mathcal{S}_2(\phi(A)|\phi(B))$. 
	\medskip
	
	To resolve this, consider $S_{\mathcal{R}, B}=\{X\in L^1(\mathcal{M}): \mathcal{R}\circ\phi(X)=X\}$. A state $A\in S_{\mathcal{R}, B}$ is said to be {\it recoverable} and $S_{\mathcal{R}, B}$ can be identified with the space of all normal, recoverable functionals on $\mathcal{M}$. We show that there exists a channel $\psi:L^1(\mathcal{M})\to L^1(\mathcal{M})$ such that $\psi(\mathcal{M})\subset\mathcal{M}$, $\mathrm{im}\, \psi=S_{\mathcal{R}, B}$ and $(\mathcal{R}\circ\phi)^n\to\psi$ in norm as operators on $L^p(\mathcal{M})$, $1\,\textless\, p\,\textless\,\infty$. 
	\medskip
	
	This shows that for any $\varepsilon\,\textgreater\, 0$, there exists $N$ such that for any arbitrary state $A$, as long as the $p$-norm of $A$ is bounded above by some fixed $C\,\textgreater\, 0$, we have $$||(\mathcal{R}\circ\phi)^n(A)-\widetilde{\psi}(A)||_1\leq ||(\mathcal{R}\circ\phi)^n(A)-\widetilde{\psi}(A)||_p\leq\varepsilon.$$
	\medskip
	
	Thus, Alice can take iterates $(\mathcal{R}\circ\phi)^n(A)$ to get arbitrarily close to a recoverable state. We note that such norm convergence of iterates is not really common in infinite dimensions, where it is more standard to have the weaker notion of strong operator convergence instead. This is a factor which makes our result interesting. It is to be noted though, that we have only been able to show strong operator convergence in $L^1$, unlike norm convergence in the other $L^p$ spaces, so this might be an interesting direction to pursue in future.
	\medskip
	
	Additionally, we also prove a decomposition theorem, showing that any normal state $A$ in $L^1$ can be written uniquely as $$A=A_0 + C$$ where $A_0\in S_{\mathcal{R}, B}$ is a recoverable state and $C$ is self-adjoint with $(\mathcal{R}\circ\phi)^n(C)\to 0$ in $L^1$. Hence, any state can be decomposed uniquely into a recoverable part and a part which vanishes in the limit under the orbit action of $\mathcal{R}\circ\phi$.
	
	\section{Main results}
	
	We first need to describe a completely positive map $\phi:L^2(\mathcal{M},\tau)\to L^2(\mathcal{N},\tau^{\prime})$. Note that for any $n$, $M_n(\mathcal{M})$ is equipped with the trace \[\tau_n\begin{pmatrix}X_{11} & \dots & X_{1n}\\\vdots & &\vdots\\X_{n1} & \dots & X_{nn}\end{pmatrix}=\sum_j \tau(X_{jj}).\] Given any $\{X_{ij}\}_{1\leq i,j\leq n}$ in $L^2(\mathcal{M})$, consider $L^2$-Cauchy sequences $(X_{ij}^{m})_m$ in $\mathcal{M}$ such that for each $i,\,j$, $X_{ij}^m\to X_{ij}$ as $m\to\infty$. Then the sequence of matrices \[\tilde{X}_m=\begin{pmatrix}X_{11}^m & \dots & X_{1n}^m\\\vdots & &\vdots\\X_{n1}^m & \dots & X_{nn}^m\end{pmatrix}\] is Cauchy in $L^2(M_n(\mathcal{M}))$, and thus, has a limit, which we denote by the block matrix \[\begin{pmatrix}X_{11} & \dots & X_{1n}\\\vdots & &\vdots\\X_{n1} & \dots & X_{nn}\end{pmatrix}.\] Conversely, any element in $L^2(M_n(\mathcal{M}))$ can be represented this way, so we identify $$M_n(L^2(\mathcal{M})):=L^2(M_n(\mathcal{M})).$$ Now, notice that any linear map $\phi:L^2(\mathcal{M})\to L^2(\mathcal{N})$ can be lifted to $\phi_n: M_n(L^2(\mathcal{M}))\to M_n(L^2(\mathcal{M}))$ by \[\phi_n\begin{pmatrix}X_{11} & \dots & X_{1n}\\\vdots & &\vdots\\X_{n1} & \dots & X_{nn}\end{pmatrix}=\begin{pmatrix}\phi(X_{11}) & \dots & \phi(X_{1n})\\\vdots & &\vdots\\\phi(X_{n1}) & \dots & \phi(X_{nn})\end{pmatrix}.\] We say such a map $\phi$ is {\it completely positive} if for all $n$, $\phi_n$ maps positive elements of $L^2(M_n(\mathcal{M}))$ to $L^2(M_n(\mathcal{N}))$. 
	\medskip
	
	It was proved in \cite{sap1} that if $\phi:\mathcal{M}\to\mathcal{N}$ is completely positive and trace preserving, then it extends continuously to $L^2(\mathcal{M},\tau)$ and the adjoint $\phi^*:L^2(\mathcal{N})\to L^2(\mathcal{M})$ is completely positive and unital. Furthermore, if $\phi$ is strictly CPTP, and $B\in\mathcal{M}$ is positive, invertible with $\tau(B)=1$, the associated Petz map $\mathcal{R}:L^2(\mathcal{N})\to L^2(\mathcal{M})$ given by $\eqref{e2}$ is also CPTP. Infact, as we will soon see, both $\phi^*$ and $\mathcal{R}$ map $\mathcal{N}$ to $\mathcal{M}$.
	\medskip
	
	We begin with a basic lemma, which will be useful throughout.
	\medskip
	
	\begin{lem}\label{l1}Let $\phi:L^2(\mathcal{M})\to L^2(\mathcal{N})$ be a completely positive map such that $\phi(I)=I$. Then, $\phi(X)\in\mathcal{N}$ whenever $X\in\mathcal{M}$ and $||\phi(X)||\leq ||X||$ for all $X\in\mathcal{M}$.\end{lem}
	\begin{proof}Assume $H$ is Hermitian and consider the spectral decomposition $$\phi(H)=\int_{\mathbb{R}}\lambda\, dP(\lambda).$$ Let $P_n=P[-n,n]$. By complete positivity, $$\begin{pmatrix}I & \phi(H)\\ \phi(H) & I\end{pmatrix}$$ is positive in $L^2(M_2(\mathcal{N}))$. Multiplying from left and right with $$\begin{pmatrix} P_n & 0\\ 0 & P_n\end{pmatrix}$$ we get $$O\leq\begin{pmatrix}P_n & \int_{-n}^n\lambda\,dP(\lambda)\\\int_{-n}^n\lambda\,dP(\lambda) & P_n\end{pmatrix}\leq \begin{pmatrix}I & \int_{-n}^n\lambda\,dP(\lambda)\\\int_{-n}^n\lambda\,dP(\lambda) & I\end{pmatrix}.$$ Thus, $$||\int_{-n}^n\lambda\, dP(\lambda)||\leq 1$$ for all $n$, so $\phi(H)$ is bounded. Any $X\in\mathcal{M}$ can be written as $X=H+iK$ where $H, K$ are Hermitian, so $\phi$ maps $\mathcal{M}\to\mathcal{N}$. Any unital completely positive map between $C^*$ - algebras is a contraction, so $||\phi(X)||\leq 1$ for all $X\in\mathcal{M}$.\end{proof}
	\medskip
	
	We get the following easy corollaries :
	\medskip
	
	\begin{cor}\label{c1} Let $\phi:L^2(\mathcal{M})\to L^2(\mathcal{M})$ be completely positive. Assume there exists a positive, invertible $B\in\mathcal{M}$ such that $\phi(B)$ is positive and invertible in $\mathcal{N}$. Then $\phi$ maps $\mathcal{M}$ to $\mathcal{N}$ and is strictly positive.\end{cor}
	\begin{proof}Apply lemma \ref{l1} to the unital CP map $$X\to \phi(B)^{-1/2}\phi(B^{1/2}XB^{1/2})\phi(B)^{-1/2}$$ and use strict positivity of unital maps. \end{proof}
	\medskip
	
	\begin{cor}\label{c2} Let $\phi:\mathcal{M}\to\mathcal{N}$ be a strict CPTP map and let $B\in\mathcal{M}$ be positive and invertible with $\tau(B)=1$. Then the Petz map $\mathcal{R}$ is strictly CPTP from $\mathcal{N}\to\mathcal{M}$.\end{cor}
	\begin{proof}Note that $\mathcal{R}(\phi(B))=B$ and apply corollary \ref{c1}.\end{proof}
	
	Let $B\in\mathcal{M}$ be positive and invertible with $\tau(B)=1$. Let $p\geq 1$. Define the Araki-Masuda $p$-norm on $L^p(\mathcal{M})$ by $||X||_{B,p}=||B^{-1/2q}XB^{-1/2q}||_p$ for all $X\in L^p(\mathcal{M})$, where $q=\frac{p}{p-1}$. If $A\in L^p(\mathcal{M})$ is positive with $\tau(A)=1$, $\mathcal{S}_p(A|B)=||A||_{B,p}$.
	\medskip
	
	\begin{lem}\label{l2}The norms $||.||_{B,p}$ and $||.||_p$ are equivalent on $L^p$.\end{lem}
	\begin{proof}Note that $||X||_{B,p}=||B^{-1/2q}XB^{-1/2q}||_p\leq ||B^{-1/q}||\,||X||_p.$ The other side is similar.\end{proof}
	It follows from Jen\v{c}ov\'a \cite{jencova} that the spaces $L^p$ equipped with the corresponding norms $||.||_{B,q}$ are {\it interpolation spaces}. To elaborate, they satisfy the following version of Riesz-Thorin interpolation  :
	\medskip
	
	\begin{prop}\label{p1}Let $(\mathcal{M},\tau)$ and $(\mathcal{N},\tau^{\prime})$ be tracial von-Neumann algebras. Let $B\in\mathcal{M}$ and $C\in\mathcal{N}$ be positive and invertible, with $\tau(B)=\tau^{\prime}(C)=1$. Let $T:L^1(\mathcal{M})\to L^1(\mathcal{N})$ be a linear operator and let $p_0,\,p_1\,q_0\,q_1\geq 1$ such that $T$ maps $L^{p_0}$ to $L^{q_0}$ and $L^{p_1}$ to $L^{q_1}$. Assume there exist $M_0,\,M_1\,\textgreater\, 0$ such that $$||T(X)||_{C,q_0}\leq M_0\,||X||_{B, p_0}$$ for all $X\in L^{p_0}$ and $$||T(X)||_{C, q_1}\leq M_1\,||X||_{B, q_1}$$ for all $X\in L^{p_1}$. Then for all $\theta\in (0,1)$, $$||T(X)||_{C, q_{\theta}}\leq M^{1-\theta}_0 M^{\theta}_1 \,||X||_{B, p_{\theta}}$$ where $\frac{1}{p_{\theta}}=\frac{1-\theta}{p_0}+\frac{\theta}{p_1}$, $\frac{1}{q_{\theta}}=\frac{1-\theta}{q_0}+\frac{\theta}{q_1}$ and $X\in L^{p_{\theta}}$.\end{prop}
	\medskip
	
	An immediate consequence is the DPI for $\mathcal{S}_p$ :
	\medskip
	
	\begin{cor}\label{c3}Let $\phi:\mathcal{M}\to\mathcal{N}$ be a strict CPTP map and let $B\in\mathcal{M}$ be positive and invertible with $\tau(B)=1$. Let $p\geq 1$. $\phi$ extends naturally to a bounded $CPTP$ map from $L^p(\mathcal{M})$ to $L^p(\mathcal{N})$, and for all $X\in L^p(\mathcal{M})$, $||\phi(X)||_{\phi(B),p}\leq ||X||_{B,p}$.  In particular, for some positive $A\in L^p(\mathcal{M})$ with $\tau(A)=1$, $\mathcal{S}_p(\phi(A)|\phi(B))\leq\mathcal{S}_p(A|B)$.\end{cor}
	\begin{proof}Suffices to show that $\phi$ is a contraction from $(L^p,||.||_{B,p})$ to $(L^p,||.||_{\phi(B),p})$. Note that $||X||_{B,1}=||X||_1$ and $||X||_{B,\infty}=||B^{-1/2}XB^{-1/2}||$. Let $X\in\mathcal{M}$. Then $$||\phi(X)||_{\phi(B),1}=||\phi(X)||_1\leq||X||_1=||X||_{B,1}.$$ Again, using the fact that the map $X\to\phi(B)^{-1/2}\phi(B^{1/2}XB^{1/2})\phi(B)^{-1/2}$ is unital CP, we get $$||\phi(X)||_{\phi(B),\infty}\leq||X||_{B,\infty}.$$ Proposition \ref{p1} then completes the proof.\end{proof}
	\medskip
	
	\begin{rem2}The idea behind the above proof is not new. It is adapted from Beigi \cite{sb} and Jen\v{c}ov\'a \cite{jencova}. An alternate proof of the DPI for $\mathcal{S}_2$, without relying on interpolation theory, can be found in \cite{sap1}.\end{rem2}
	\medskip
	
	Let $B\in\mathcal{M}$ be positive and invertible. Consider the inner product $\langle X, Y\rangle_{B}=\tau(B^{-1/2}Y^*B^{-1/2}X)$ for all $X, Y\in L^2(\mathcal{M})$. The norm induced by this is precisely $||.||_{B,2}$. Let $\phi:\mathcal{M}\to\mathcal{N}$ be strictly CPTP. Then, $\phi$ lifts to a Hilbert space contraction from $(L^2(\mathcal{M}), ||.||_{B,2})$ to $(L^2(\mathcal{N}),||.||_{\phi(B),2})$. The adjoint is given by the Petz map $\mathcal{R}$ (see $\eqref{e2}$). Note that the map $\mathcal{R}\circ\phi$ is a positive contraction on $(L^2(\mathcal{M},||.||_{B,2})$. Also, given a linear map $T:(L^p,||.||_{B,p})\to (L^p,||.||_{B,p})$, denote its operator norm by $||T||_{B,p}$. We can now state and prove our first theorem.
	\medskip
	
	\begin{thm}\label{t1} Let $B\in\mathcal{M}$ be positive and invertible with $\tau(B)=1$. Let $\phi:\mathcal{M}\to\mathcal{N}$ be a strict CPTP map. Then there exists a CPTP map $\psi:L^2(\mathcal{M})\to L^2(\mathcal{M})$ such that $\psi(X)\in\mathcal{M}$ for all $X\in\mathcal{M}$, $\mathcal{R}\circ\phi\circ\psi=\psi$ and $(\mathcal{R}\circ\phi)^n\to\psi$ in the operator norm on $B(L^2(\mathcal{M}))$. \end{thm}
	\begin{proof}$\mathcal{R}\circ\phi$ is a positive contraction on $(L^2(\mathcal{M},||.||_{B,2})$. Let $\mathcal{V}=\{X\in L^2:\mathcal{R}\circ\phi(X)=X\}$. Since $\mathcal{R}\circ\phi(B)=B$, $\mathcal{V}\neq\{0\}$. Let $\psi$ be the orthogonal projection on $\mathcal{V}$ with respect to the inner product $\langle\underline{\hspace{3mm}},\underline{\hspace{3mm}}\rangle_{B}$. We then have the block matrix representations with respect to $\mathcal{V}\oplus\mathcal{V}^{\perp}$ : \[\mathcal{R}\circ\phi = \begin{pmatrix} I & O\\ O & C\end{pmatrix} \] and \[\psi=\begin{pmatrix} I & O\\ O & O\end{pmatrix}\] where $C$ is a positive operator on $\mathcal{V}^{\perp}$. Since $\mathcal{R}\circ\phi$ is a contraction and $\mathcal{V}\neq\{0\}$, $1$ is an isolated point of $\sigma(\mathcal{R}\circ\phi)$. Hence there exists $0\,\leq\delta\,\textless\,1$ such that $C\leq\delta$. Thus, $$||(\mathcal{R}\circ\phi)^n-\psi||_{B,2}\leq \delta^n\to 0$$ as $n\to\infty$. Since $||.||_2$ and $||.||_{B,2}$ are equivalent on $L^2$ (see lemma \ref{l2}), $(\mathcal{R}\circ\phi)^n\to\psi$ in norm on $B(L^2)$.
	\medskip
	
	Since $(\mathcal{R}\circ\phi)^n$ is CPTP for all $n$, $\psi$ is also CPTP. Being the orthogonal projection onto $\mathcal{V}$, it naturally satisfies $\mathcal{R}\circ\phi\circ\psi=\psi$. Also, by corollary \ref{c1}, $\psi$ maps $\mathcal{M}$ to $\mathcal{M}$. \end{proof}
	\medskip
	
	Since $\psi$ is CPTP, it is a contraction with respect to $||.||_1$ and extends uniquely to $L^1$.
	\begin{cor}\label{c4}$(\mathcal{R}\circ\phi)^n\to\psi$ in the strong operator topology on $L^1$.\end{cor}
	\begin{proof}$(\mathcal{R}\circ\phi)^n$ is a sequence of $L^1$-contractions, converging strongly to $\psi$ on the dense subspace $L^2$. Therefore, it converges strongly in $L^1$.\end{proof}
	
	An important consequence of Theorem \ref{t1} is our next result :
	\medskip
	
	\begin{thm}\label{t2}$(\mathcal{R}\circ\phi)^n\to\psi$ in norm as operators on $L^p$ for all $1\,\textless\,p\,\textless\,\infty$.\end{thm}
	
	\begin{proof}We show that $||(\mathcal{R}\circ\phi)^n-\psi||_{B,p}\to 0$. The case $p=2$ is done. Let $p\in(1,2)$. Note that since $\mathcal{R}\circ\phi$ and $\psi$ are CPTP, corollary \ref{c3} implies $$||(\mathcal{R}\circ\phi)^n-\psi||_{B,1}\leq ||(\mathcal{R}\circ\phi)^n||_{B,1}+||\psi||_{B,1}= 2.$$ By Proposition \ref{p1}, $$||(\mathcal{R}\circ\phi)^n-\psi||_{B,p}\leq 2^{\frac{2}{p}-1}||(\mathcal{R}\circ\phi)^n-\psi||^{2(1-\frac{1}{p})}_{B,2}\to 0$$ as $n\to\infty$. The case $p\in(2,\infty)$ is similar. Now use equivalence of $||.||_p$ and $||.||_{B,p}$ to complete the proof.\end{proof}
	\medskip
	
	\begin{rem}\end{rem}
	\begin{enumerate}
	\item The sequence of linear maps $(\mathcal{R}\circ\phi)^n-\psi$ map $\mathcal{M}$ to $\mathcal{M}$ by corollary \ref{c1}, but it need not be uniformly bounded. However, when the operator norm on $\mathcal{M}$ is replaced with $||.||_{B,\infty}$, it is uniformly bounded, which allows us to use complex interpolation and conclude norm convergence in $L^p$.
	\medskip
	
	\item Theorem \ref{t2} demonstrates how iterates of the Petz map composed with the channel can be used to take any state arbitrarily close to a recoverable state. This has potential applications in quantum information and communication. We do not know whether the operator norm convergence holds in $L^1$. Future research may shed some more light on this problem. However, as demonstrated in corollary \ref{c4}, strong operator convergence still holds.\end{enumerate}
	\medskip
	
	We will now prove and interesting decomposition theorem in $L^1$, which will allow us to uniquely write any normal state $A$ as the sum of a recoverable state $A_0$ and a self-adjoint element  $C$ such that $(\mathcal{R}\circ\phi)^n(C)\to 0$. 
	\medskip
	
	Let $\psi$ be the channel defined in the proof of Theorem \ref{t1}. Extend $\psi$ to $L^1(\mathcal{M})$ and denote, by abuse of notation, the extension as $\psi$ as well. Let $$S_{\mathcal{R}, B}=\{X\in L^1(\mathcal{M}): \mathcal{R}\circ\phi(X)=X\}$$ be the subspace of all recoverable elements. We show that $\psi$ projects onto $S_{\mathcal{R},B}$.
	\medskip
	
	\begin{lem}\label{l3}$\psi:L^1(\mathcal{M})\to L^1(\mathcal{M})$ satisfies $\psi^2=\psi$ and $\mathrm{im}\,\psi=S_{\mathcal{R},B}$.\end{lem}
	\begin{proof} By corollary \ref{c4}, $(\mathcal{R}\circ\phi)^n(X)\to\psi(X)$ for all $X\in L^1$. Let $X_0\in S_{\mathcal{R},B}$. Then $$(\mathcal{R}\circ\phi)^n(X_0)=X_0$$ for all $n$, and therefore, $\psi(X_0)=X_0$. Again, for any $X$, $$(\mathcal{R}\circ\phi)(\psi(X))=\lim_{n\to\infty}(\mathcal{R}\circ\phi)^{n+1}(X)=\psi(X).$$ Thus, $\mathrm{im}\,\psi=S_{\mathcal{R}, B}$ and $\psi$ fixes $S_{\mathcal{R},B}$.\end{proof}
	\medskip
	
	\begin{thm}\label{t3}Let $A\in L^1(\mathcal{M})$ be positive with $\tau(A)=1$. Then there exists a positive $A_0\in S_{\mathcal{R},B}$ with $\tau(A_0)=1$ and a self-adjoint $C\in L^1(\mathcal{M})$ with $||(\mathcal{R}\circ\phi)^n(C)||_1\to 0$ and $A=A_0+C$. This decomposition is unique. \end{thm}
	\begin{proof}By lemma \ref{l3}, $\psi:L^1(\mathcal{M})\to L^1(\mathcal{M})$ is a projection onto $S_{\mathcal{R},B}$. Hence, \[L^1(\mathcal{M})\cong S_{\mathcal{R},B}\oplus\ker\,\psi\label{e1}\tag{1}\] as Banach spaces. Set $A_0=\psi(A)$ and $C=A-\psi(A)$. Since $\psi$ is a channel, $A_0$ is a state. By corollary \ref{c4}, $$(\mathcal{R}\circ\psi)^n(C)\to 0.$$ Uniqueness of the decomposition follows from $\eqref{e1}$.\end{proof}
	\medskip
	
	\begin{rem2}We have proved some interesting structural results for normal states on a von-Neumann algebra from the perspective of recoverability. It is to be noted that the norms $||.||_{B,p}$, from which the sandwiched quasi-relative entropies $\mathcal{S}_p$ arise, played an integral part in the proofs, even though the statements of the results themselves have nothing to do with $\mathcal{S}_p$. The question of norm convergence of $(\mathcal{R}\circ\phi)^n$ as operators on $L^1(\mathcal{M})$ is left open, we hope for some future progress on this. One possible direction is to look at the spectral radius of $\mathcal{R}\circ\phi$, when restricted to the invariant subspace $\ker\,\psi$.  \end{rem2}
	\medskip
	
	\begin{acknowledgement}
	The author thanks the Indian Statistical Institute for supporting his PhD financially.
	\end{acknowledgement}
	
	\textbf{Conflict of interest :} The author declares no conflict of interest.
	\medskip
	
	\textbf{Data availability statement :} No available data has been used.

\end{document}